\documentclass[12pt]{article}
\usepackage[utf8]{inputenc}
\usepackage[english]{babel}
\usepackage{amsmath}
\usepackage{amsthm,amssymb}
\usepackage{bm}
\usepackage{amsfonts}
\usepackage[margin=30mm,includefoot]{geometry}
\usepackage{fancyvrb}
\usepackage{dsfont}
\usepackage{mathtools}
\usepackage{mathrsfs}
\usepackage{float}
\usepackage{caption}
\usepackage{pgfgantt}
\usepackage{indentfirst}
\usepackage{tikz}
\usepackage[font=small,labelfont=bf]{caption}
\usepackage[nottoc]{tocbibind}
\usepackage[toc]{appendix}
\usepackage{algorithm}
\usepackage{multirow}
\usepackage{titlesec}
\usepackage[affil-it]{authblk}
\newcommand{\email}[1]{\texttt{\small #1}}

\newtheorem{proposition}{Proposition}

\author[1]{Yiping Guo \footnote{Email: \email{y246guo@uwaterloo.ca} (Yiping Guo)}}
\author[2]{Howard Bondell \footnote{Email: \email{howard.bondell@unimelb.edu.au} (Howard Bondell)}}
\affil[1]{\normalsize Department of Statistics and Actuarial Science, University of Waterloo}
\affil[2]{\normalsize School of Mathematics and Statistics, University of Melbourne}

\usepackage[round]{natbib}
\usepackage{titlesec}

\title{On Robust Probabilistic Principal Component Analysis using Multivariate \textit{t}-Distributions}
\date{}

\allowdisplaybreaks
\begin{document}
\maketitle

\begin{abstract}
	\noindent Probabilistic principal component analysis (PPCA) is a probabilistic reformulation of principal component analysis (PCA), under the framework of a Gaussian latent variable model. To improve the robustness of PPCA, it has been proposed to change the underlying Gaussian distributions to multivariate \textit{t}-distributions. Based on the representation of \textit{t}-distribution as a scale mixture of Gaussian distributions, a hierarchical model is used for implementation. However, in the existing literature, the hierarchical model implemented does not yield the equivalent interpretation. 
	
	In this paper, we present two sets of equivalent relationships between the high-level multivariate $t$-PPCA framework and the hierarchical model used for implementation. In doing so, we clarify a current misrepresentation in the literature, by specifying the correct correspondence. In addition, we discuss the performance of different multivariate $t$ robust PPCA methods both in theory and simulation studies, and propose a novel Monte Carlo expectation-maximization (MCEM) algorithm to implement one general type of such models. 
\end{abstract}
\textbf{Keywords}: 
Heavy tailed regression; Latent variable; Principal components; Robust estimation.

\section{Introduction}
Principal component analysis (PCA) is a powerful multivariate statistical technique for dimensionality reduction, which is widely used in many areas such as image processing and pattern recognition. Originally, PCA finds the best linear projection of a high-dimensional data set onto a lower-dimensional space using eigendecomposition or singular value decomposition, such that the reconstruction error is minimized. However, since traditional PCA is a deterministic approach, one limitation is the absence of a probabilistic framework which can be then used in conjunction with likelihood-based or Bayesian approaches.

\cite{tipping1999probabilistic} proposed the concept of probabilistic principal component analysis (PPCA), which reformulates the PCA problem from a Gaussian latent variable model point of view. Denote the PCA score vector in the \textit{d}-dimensional latent space as $\bm{z}$, we assume that it follows a standard normal distribution, that is, $\bm{z}\sim N(\bm{0},\bm{I})$. Given the PCA score, we also assume the data vector $\bm{x}$ in the \textit{q}-dimensional space typically with $q\gg d$, is also normally distributed, that is, $\bm{x}|\bm{z}\sim N(\bm{W}\bm{z}+\bm{\mu},\sigma^2 \bm{I})$. After observing $\bm{x}$, PPCA then estimates all unknown parameters $\bm{W},\bm{\mu},\sigma^2$ using maximum likelihood method. \cite{tipping1999probabilistic} showed that the maximum likelihood estimator (MLE) of the transformation matrix $\bm{W}$ spans the principal space, and the estimated posterior mean of the latent vector can be treated as the ``principal components". \cite{tipping1999probabilistic} also proved that when $\sigma^2 \to 0$, the posterior mean reduces to the orthogonal projection of the observed data points onto the latent space, and so it recovers the conventional PCA model. Although the analytical solution to the standard PPCA exists, an alternative is to use the expectation-maximization (EM) algorithm \citep{dempster1977maximum}. One advantage of using the EM algorithm is to avoid working directly with the sample covariance matrix, which is computationally intensive for high dimensional data. Another advantage is that it can generalise to the case that we choose to make modifications to the underlying distributions of the data $\bm{x}|\bm{z}$ and/or the latent variable $\bm{z}$.  

It is well known that the traditional PCA lacks robustness because the structure of sample variance matrix can be highly influenced by extreme values. Many researchers have worked on robust PCA but no approach is perfect under all situations. A classical category of approaches is to construct a robust sample covariance matrix, using the concept of Mahalanobis distance \citep{campbell1980robust}. \cite{candes2011robust} proposed to decompose the sample covariance matrix into a low-rank component (for nonoutliers) and a sparse component (for outliers), and use projection pursuit to recover both components by convex optimization.  

Similarly, PPCA is also highly sensitive to outliers since the Gaussian assumptions are used for both data and latent variables. The first robust PPCA model was proposed by \cite{archambeau2006robust}, which changes the conditional distribution of the data $\bm{x}|\bm{z}$ and the latent variable $\bm{z}$ from Gaussian to the multivariate \textit{t}-distributions. Because \textit{t}-distributions have heavier tails than normal distributions, the resulting maximum likelihood estimates will be more robust to outliers. To the best of our knowledge, all following papers working on robust PPCA used the same modification. \cite{archambeau2008mixtures} developed mixtures of robust PPCA model to work on high-dimensional nonlinear data by combining locally linear models, and this model was later applied to process monitoring by  \cite{zhu2014robust}. \cite{gai2008robust} extended this idea to robust Bayesian PCA and provided both an EM algorithm and the variational Bayes approach for implementation. 
Furthermore, \cite{chen2009robust} discussed the issue of missing data using this robust PPCA model based on multivariate \textit{t}-distributions.  

Another advantage of using \textit{t}-distributions is that they can be expressed as continuous scale mixtures of normal distributions \citep{liu1995ml}. From this, we can avoid dealing with complicated density function of multivariate \textit{t}-distributions, but instead introduce a new latent gamma distributed scale variable to obtain simpler conditional distributions. Nonetheless, the derivation of the EM algorithm is still more complicated than the standard PPCA. To simplify the derivations, \cite{archambeau2006robust} and all of the following literature used the conjugacy properties of the exponential family by introducing the same scale variable to both the data $\bm{x}|\bm{z}$ and the latent variable $\bm{z}$. However, as we will show in this paper, the hierarchical model using the same scale variable is not equivalent to the original motivated model, which means the simplification that has been used in the existing literature does not match the original distributional model. Conceptually, there is no particular reason to introduce highly dependent robustness simultaneously to both the data space and the latent space, other than for a simplification in computation. This will implicitly assume that if a data point is considered to be an outlier in the data space, it is also considered to be an outlier in the lower-dimensional latent space. However, as we show in this manuscript, the construction that has been used in the past, actually refers to different model constructions with different statistical interpretations. If we instead use other robust PPCA models which will be introduced in this paper, the conjugacy properties will not hold for all steps. As a consequence, there will be no closed-form expression for the posterior means, so Monte Carlo methods might need to be involved in the E-step. 

In the following sections, we first review the scale mixture representation of multivariate $t$-distributions and standard robust PPCA methods. Then two general types of robust PPCA methods based on multivariate $t$-distributions will be studied in detail. Furthermore, an Monte Carlo expectation-maximization (MCEM) algorithm for implementing one type of such models will be derived. Finally, two simulation studies are discussed to evaluate and compare different models.

% to evaluate the developed robust model 
% This paper will start from a quick review of the scale mixture representation of multivariate-$t$ distributions in Section 2
% In this paper, we will first quickly multivariate $t$-distributions and the normal-gamma conjugacy properties in Section 2. Next, we propose a set of equivalent robust PPCA models and make comparisons in Section 3. Then, we derive a Monte Carlo Expectation-Maximization algorithm (MCEM) for a general robust PPCA model in Section 4. Section 5 presents a simulation study based on the models discussed before. Lastly, we give a conclusion to our work in Section 6. 

\section{Background}
\subsection{Multivariate $t$-Distributions}
Let a $q$-dimensional random vector $\bm{x}$ follow a multivariate \textit{t}-distribution $\bm{t}_\nu(\bm{\mu},\bm{\Sigma})$, with mean vector $\bm{\mu}$, scale matrix $\bm{\Sigma}$, and degrees of freedom $\nu$. Note that the covariance matrix in this case is given by $\frac{\nu}{\nu-2}\bm{\Sigma}$ instead of $\bm{\Sigma}$ when $\nu>2$. The corresponding probability density function is \citep{kibria2006short}:
\begin{equation}
    f(\bm{x})=
\frac{\Gamma[(\nu+q) / 2]}{\Gamma(\nu / 2) \nu^{q / 2} \pi^{q / 2}|\mathbf{\Sigma}|^{1 / 2}}\left[1+\frac{1}{\nu}(\bm{x}-\boldsymbol{\mu})^{T} \mathbf{\Sigma}^{-1}(\bm{x}-\boldsymbol{\mu})\right]^{-(\nu+q) / 2}.
\end{equation}
We can easily see that this density function decays at a polynomial rate, which indicates that \textit{t}-distributions have heavier tails than normal distributions since normal densities decay exponentially. 

One way to represent a multivariate \textit{t}-distribution $\bm{t}_{\nu}(\bm{\mu},\Sigma)$ is to express it as a scale mixture of normal distributions \citep{liu1995ml}:
\begin{equation}
    \begin{gathered}
    \bm{x}|u \sim N\left(\bm{\mu},\frac{\bm{\Sigma}}{u}\right),\\
    u \sim \mathrm{Ga}\left(\frac{\nu}{2},\frac{\nu}{2}\right).
    \end{gathered}
\end{equation}
Another important property related to multivariate \textit{t}-distributions is the conjugacy between the gamma prior and the normal likelihood. Given a \textit{q}-dimensional normal likelihood $\bm{x}|u \sim N(\bm{\mu},\frac{\bm{\Sigma}}{u})$ and a gamma prior $u  \sim \mathrm{Ga}(\alpha,\beta)$, we have:
\begin{equation}
    \begin{gathered}
    u|\bm{x} \sim \mathrm{Ga}\left(\alpha+\frac{q}{2}, \beta+\frac{(\bm{x}-\boldsymbol{\mu})^{T} \mathbf{\Sigma}^{-1}(\bm{x}-\boldsymbol{\mu})}{2}\right),\\
    \bm{x} \sim t_{2\alpha}\left(\bm{\mu},\frac{\beta}{\alpha}\bm{\Sigma}\right),
    \end{gathered}
\end{equation}
where one can easily see that (2) is a special case of (3) with $\alpha=\beta=\frac{\nu}{2}$.

\subsection{Probabilistic Principal Component Analysis (PPCA)}
\cite{tipping1999probabilistic} proposed the standard probabilistic principal component analysis (PPCA) method, which is a generative Gaussian latent variable model. Firstly, a $d$-dimensional latent variable $\bm{z}$ is generated from an isotropic Gaussian distribution with zero mean and unit variance for each component:
\begin{equation}
    \bm{z} \sim N(\bm{0},\bm{I}).
\end{equation}
Then, conditioning on the ``latent score" $\bm{z}$, a $q$-dimensional ``data vector" $(q\geq d)$ is generated from another isotropic Gaussian distribution:
\begin{equation}
    \bm{x}|\bm{z} \sim N(\bm{W}\bm{z}+\bm{\mu},\sigma^2\bm{I}),
\end{equation}
where the maximum likelihood estimate of $\bm{W}$ spans the $q$-dimensional principal subspace, and the estimated posterior mean $\mathbb{E}(\bm{z}|\bm{x})$ will recover the conventional PCA as $\sigma^2 \rightarrow 0$ \citep{tipping1999probabilistic}.

\section{Robust PPCA Models}
Since the standard PPCA under the Gaussian framework is sensitive to the existence of outliers, \cite{archambeau2006robust} proposed a robust PPCA model (6) by replacing the Gaussian distributions in (4) and (5) by multivariate $t$-distributions with the same degrees of freedom:
\begin{equation}
    \begin{gathered}
        \bm{x}|\bm{z} \sim     t_{\nu}(\bm{W}\bm{z}+\bm{\mu},\sigma^2\bm{I}),\\
        \bm{z} \sim t_{\nu}(\bm{0},\bm{I}).
    \end{gathered}
\end{equation}
To simplify the derivation of the corresponding EM algorithm, \cite{archambeau2006robust} presented a hierarchical model (7) by expressing the $t$-distributions as a scale mixture of normal distributions:
\begin{equation}
    \begin{gathered}
        \bm{x}|\bm{z},u \sim N\left(\bm{W}\bm{z}+\bm{\mu},\frac{\sigma^2\bm{I}}{u}\right),\\
        \bm{z}|u \sim N\left(\bm{0},\frac{\bm{I}}{u}\right),\\
        u \sim \mathrm{Ga}\left(\frac{\nu}{2},\frac{\nu}{2}\right).
    \end{gathered}
\end{equation}
The use of a single latent scale variable $u$, results in simpler computation. However, it invokes dependence between the variance in the data space and the latent space, and after integrating out the scale variable $u$, the dependence between $\bm{x}$ and $\bm{z}$ shows up in the conditional variance as well as the conditional mean and hence it follows that (6) and (7) are not equivalent, as previously claimed. Details are shown in Section 3.1.2. 

We noted that a number of papers (\cite{archambeau2008mixtures},\cite{gai2008robust},\cite{chen2009robust} and \cite{zhu2014robust}, for example) use this equivalence and build upon it further. We will now discuss the equivalent hierarchical models for (6) and equivalent marginal model for (7).

\subsection{Model formulation}

\subsubsection{Conditional and Latent $t$-Models}

In this subsection, we will propose a general PPCA model with conditional and latent $t$-distributions. It turns out that Model (6) is a special case of this general PPCA model. We denote this general type of models (8) as ``C\&L $t$-model" for future reference. The corresponding directed acyclic graph (DAG) and model equivalence are:
\begin{center}
    \begin{tikzpicture}[nodes={draw, circle}, ->]

    \node[text centered] (u2) {$u_2$};
    \node[right = 2 of u2, text centered] (z) {$\bm{z}$};
    \node[right = 2 of z, text centered] (x) {$\bm{x}$};
    \node[above = 1.5 of z, text centered] (u1) {$u_1$};
     
    \draw[->, line width= 1] (u2) --  (z);
    \draw[->, line width= 1] (u1) --  (x);
    \draw [->, line width= 1] (z) -- (x);
    
    \end{tikzpicture}
\end{center}

\begin{equation}
    \left\{\begin{aligned} 
      &\bm{x}|\bm{z} \sim t_{\nu_1}(\bm{W}\bm{z}+\bm{\mu},\sigma^2\bm{I}), \\ 
      &\ \bm{z} \ \, \sim t_{\nu_2}(\bm{0},\bm{I}),
    \end{aligned}\right. \iff 
    \left\{\begin{aligned}
      &\bm{x}|\bm{z},u_1 \sim N\left(\bm{W}\bm{z}+\bm{\mu},\frac{\sigma^2\bm{I}}{u_1}\right),\\
    &\ \bm{z}|u_2 \ \ \sim N\left(\bm{0},\frac{\bm{I}}{u_2}\right),\\ 
    &\ \ u_1 \quad \, \sim \mathrm{Ga}\left(\frac{\nu_1}{2},\frac{\nu_1}{2}\right),\\
    &\ \ u_2 \quad \,  \sim \mathrm{Ga}\left(\frac{\nu_2}{2},\frac{\nu_2}{2}\right),
    \end{aligned}\right.
\end{equation}
where $u_1$ and $u_2$ are independent.

We now show the equivalence between the models given above. Starting with the hierarchical model on the RHS, we can directly see that $\bm{z} \sim t_{\nu_2}(\bm{0},\bm{I})$ from the Gaussian scale mixture expression for the multivariate $t$-distributions in (2):
\begin{equation}
    p(\bm{z})=\int_0^{\infty}p(\bm{z}|u_2)p(u_2)\mathrm{d}u_2=\int_0^{\infty}N\left(\bm{0},\frac{\bm{I}}{u_2}\right)\cdot \mathrm{Ga}\left(\frac{\nu_2}{2},\frac{\nu_2}{2}\right)\mathrm{d}u_2\sim t_{\nu_2}(\bm{0},\bm{I}).
\end{equation}
On the other hand, in the Gaussian scale mixture hierarchical model, $u_1$ and $u_2$ are independent and $\bm{z}$ is completely determined by $u_2$, so $\bm{z}$ and $u_1$ are also independent, i.e, $p(u_1|\bm{z})=p(u_1)$. Then the conditional distribution $p(\bm{x}|\bm{z})$ is:
\begin{align}
    p(\bm{x}|\bm{z})=\int^\infty_0 p(\bm{x}|\bm{z},u_1)p(u_1)\mathrm{d}u_1&=\int^\infty_0 N\left(\bm{W}\bm{z}+\bm{\mu},\frac{\sigma^2\bm{I}}{u_1}\right)\cdot \mathrm{Ga}\left(\frac{\nu_1}{2},\frac{\nu_1}{2} \right)\mathrm{d}u_1\nonumber\\
    &\sim t_{\nu_1}(\bm{W}\bm{z}+\bm{\mu},\sigma^2\bm{I}),
\end{align}
where the last line follows from result (3) or directly from the Gaussian scale mixture representation for multivariate $t$-distributions.

There are two interesting special cases of C\&L $t$-models. The first one is to set the degrees of freedom for both scale variables $u_1$ and $u_2$ to be equal, which leads to (6):
\begin{equation}
    \left\{\begin{aligned} 
      &\bm{x}|\bm{z} \sim t_{\nu}(\bm{W}\bm{z}+\bm{\mu},\sigma^2\bm{I}), \\ 
      &\ \bm{z} \ \, \sim t_{\nu}(\bm{0},\bm{I}),
    \end{aligned}\right. \iff 
        \left\{
    \begin{aligned}
      &\bm{x}|\bm{z},u_1 \sim N\left(\bm{W}\bm{z}+\bm{\mu},\frac{\sigma^2\bm{I}}{u_1}\right),\\
    &\ \bm{z}|u_2 \ \ \sim N\left(\bm{0},\frac{\bm{I}}{u_2}\right),\\
    &\ \ u_1 \quad \, \sim \mathrm{Ga}\left(\frac{\nu}{2},\frac{\nu}{2}\right),\\
    &\ \ u_2 \quad \, \sim \mathrm{Ga}\left(\frac{\nu}{2},\frac{\nu}{2}\right),
    \end{aligned}    
    \right.
\end{equation}
where $u_1$ and $u_2$ are independent. Then, it is clear that the correct Gaussian scale mixture expression of (6) takes a different form from (7).

Another special case is to let $\nu_2 \rightarrow \infty$, which implies that the latent variable $\bm{z}$ follows a multivariate normal distribution. Briefly speaking, this model assumes that the extremeness only arises from the data space. More interpretations will be presented in the next subsection. We denote this special case as ``Conditional $t$-model", and present its DAG and model structure:

\begin{center}
    \begin{tikzpicture}[nodes={draw, circle}, ->]

    \node[text centered] (z) {$\bm{z}$};
    \node[above = 1.5 of z, text centered] (u) {$u$};
    \node[right = 2 of z, text centered] (x) {$\bm{x}$};
    
    \draw[->, line width= 1] (u) --  (x);
    \draw [->, line width= 1] (z) -- (x);
    
    \end{tikzpicture}
\end{center}
\begin{equation}
    \left\{\begin{aligned} 
      &\bm{x}|\bm{z} \sim t_{\nu}\left(\bm{W}\bm{z}+\bm{\mu},\sigma^2\bm{I}\right), \\ 
      &\ \bm{z} \ \ \sim N(\bm{0},\bm{I}),
    \end{aligned}\right. \iff 
    \left\{\begin{aligned}
    &\bm{x}|\bm{z},u \sim N\left(\bm{W}\bm{z}+\bm{\mu},\frac{\sigma^2\bm{I}}{u}\right),\\
    &\ \ \bm{z} \quad \, \sim N\left(\bm{0},\bm{I}\right),\\
    &\ \ u \quad \, \sim \mathrm{Ga}\left(\frac{\nu}{2},\frac{\nu}{2}\right).
    \end{aligned}\right.
\end{equation}

\subsubsection{Marginal $t$-Models}
Next, we show that the following model is equivalent to (7), and we denote it as ``Marginal $t$-model". The corresponding DAG and model equivalence are:

\begin{center}
    \begin{tikzpicture}[nodes={draw, circle}, ->]

    \node[text centered] (u) {$u$};
    \node[right = 2 of u, text centered] (z) {$\bm{z}$};
    \node[right = 2 of z, text centered] (x) {$\bm{x}$};
     
    \draw[->, line width= 1] (u) --  (z);
    \draw [->, line width= 1] (z) -- (x);
    \draw[->, line width=1] (u) to  [out=270,in=270,     looseness=0.6]  (x);
    \end{tikzpicture}
\end{center}
\begin{equation}
    \left\{\begin{aligned} 
      &\bm{x}|\bm{z} \sim t_{\nu+d}\left(\bm{W}\bm{z}+\bm{\mu},\frac{\nu+\bm{z}^T\bm{z}}{\nu+d}\sigma^2\bm{I}\right), \\ 
      &\ \bm{z} \ \, \sim t_{\nu}(\bm{0},\bm{I}),
    \end{aligned}\right. \iff 
    \left\{\begin{aligned}
    &\bm{x}|\bm{z},u \sim N\left(\bm{W}\bm{z}+\bm{\mu},\frac{\sigma^2\bm{I}}{u}\right),\\
    &\ \bm{z}|u \ \ \sim N\left(\bm{0},\frac{\bm{I}}{u}\right),\\
    &\ \ u \quad \,\sim \mathrm{Ga}\left(\frac{\nu}{2},\frac{\nu}{2}\right),
    \end{aligned}\right.
\end{equation}
where the RHS is what has been used in previous literature as the hierarchical model. Note that this is not equivalent to (6), it is instead equivalent to the LHS, where the latent variable $\bm{z}$ appears not just in the conditional mean but in the conditional variance as well.   

The derivation is essentially the same as that for the C\&L $t$-models. From the Gaussian scale mixture expression for multivariate $t$-distributions in (2), we can directly see that $\bm{z} \sim t_{\nu}(\bm{0},\bm{I})$:
\begin{equation}
    p(\bm{z})=\int_0^{\infty}p(\bm{z}|u)p(u)\mathrm{d}u=\int_0^{\infty}N\left(\bm{0},\frac{\bm{I}}{u}\right)\cdot \mathrm{Ga}\left(\frac{\nu}{2},\frac{\nu}{2}\right)\mathrm{d}u\sim t_{\nu}(\bm{0},\bm{I}).
\end{equation}
On the other hand, to determine the conditional distribution $p(\bm{x}|\bm{z})$, we need find the conditional distribution $p(u|\bm{z})$ using the normal-gamma conjugacy property (3):
\begin{align}
    p(u|\bm{z}) \propto p(\bm{z}|u)p(u)&= N\left(\bm{0},\frac{\bm{I}}{u}\right)\cdot \mathrm{Ga}\left(\frac{\nu}{2},\frac{\nu}{2}\right)\sim \mathrm{Ga}\left(\frac{\nu+d}{2},\frac{\nu+\bm{z}^T\bm{z}}{2} \right).
\end{align}
Then, the conditional distribution $p(\bm{x}|\bm{z})$ is:
\begin{align}
    p(\bm{x}|\bm{z})=\int^\infty_0 p(\bm{x}|\bm{z},u)p(u|\bm{z})\mathrm{d}u&=\int^\infty_0 N\left(\bm{W}\bm{z}+\bm{\mu},\frac{\sigma^2\bm{I}}{u}\right)\cdot \mathrm{Ga}\left(\frac{\nu+d}{2},\frac{\nu+\bm{z}^T\bm{z}}{2} \right)\mathrm{d}u\nonumber \\
    & \sim t_{\nu+d}\left(\bm{W}\bm{z}+\bm{\mu},\frac{\nu+\bm{z}^T\bm{z}}{\nu+d}\sigma^2\bm{I}\right),
\end{align}
where the last line follows from result (3) or directly from the Gaussian scale mixture expression for multivariate $t$-distributions.

We call this a Marginal $t$-model, as it can be shown that the marginal distribution of $\bm{x}$ is actually $t$-distributed in this case:
\begin{equation}
    \bm{x} \sim t_{\nu}(\bm{\mu},\bm{W}\bm{W}^T+\sigma^2\bm{I}).
\end{equation}
To show this, first we notice that the conditional distribution $p(\bm{x}|u)$ is also a normal distribution since normal prior for the mean parameter conjugates to a normal likelihood:
\begin{align}
    p(\bm{x}|u)=\int p(\bm{x}|\bm{z},u)p(\bm{z}|u)\mathrm{d}\bm{z}&=\int N\left(\bm{W}\bm{z}+\bm{\mu},\frac{\sigma^2\bm{I}}{u}\right)\cdot N\left(\bm{0},\frac{\bm{I}}{u} \right)\mathrm{d}\bm{z}\nonumber\\
    & \sim N\left(\bm{\mu},\frac{\bm{W}\bm{W}^T+\sigma^2\bm{I}}{u}\right).
\end{align}
Then using the Gaussian scale mixture expression of multivariate $t$-distributions (3), we show that the data vector $\bm{x}$ is multivariate $t$-distributed:
\begin{align}
    p(\bm{x})=\int_0^{\infty}p(\bm{x}|u)p(u)\mathrm{d}u&=\int_0^{\infty}N\left(\bm{\mu},\frac{\bm{W}\bm{W}^T+\sigma^2\bm{I}}{u}\right)\cdot \mathrm{Ga}\left(\frac{\nu}{2},\frac{\nu}{2}\right)\mathrm{d}u \nonumber\\
    & \sim t_{\nu}(\bm{\mu},\bm{W}\bm{W}^T+\sigma^2\bm{I}).
\end{align}

\subsection{Model Comparison}
In this section, we will discuss some connections between C\&L $t$-models (8) and Marginal $t$-models (13) and also make comparisons both in theory and simulations.

\subsubsection{Conditional and Latent $t$-Models}

From the DAGs, it is easy to see that the fundamental differences between these models are the scale variable $u$. For Marginal $t$-models, the same scale variable $u$ will determine both distributions $\bm{x}|\bm{z}$ and $\bm{z}$. However, no such dependence exists in C\&L $t$-models, where $u_1$ and $u_2$ separately contribute to the variance component of $\bm{x}|\bm{z}$ and $\bm{z}$. 

As we can see from the last section, C\&L $t$-models
% (with the same degrees of freedom of scale variable $u_1$ and $u_2$)
are indeed the models that past papers motivated since they change the conditional distribution of the data vector $p(\bm{x}|\bm{z})$ and marginal distribution of latent variable $\bm{z}$ from multivariate Gaussian distributions to multivariate $t$-distributions. In this case, the marginal distribution of the data vector $\bm{x}$ is also not a multivariate $t$-distribution. To show this, we first notice that the conditional distribution $p(\bm{x}|u)$ is also a normal distribution since normal prior for the mean parameter conjugates to a normal likelihood:
\begin{align}
    p(\bm{x}|u_1,u_2)=\int p(\bm{x}|\bm{z},u_1)p(\bm{z}|u_2)\mathrm{d}\bm{z}&=\int N\left(\bm{W}\bm{z}+\bm{\mu},\frac{\sigma^2\bm{I}}{u_1}\right)\cdot N\left(\bm{0},\frac{\bm{I}}{u_2} \right)\mathrm{d}\bm{z}\nonumber\\
    & \sim N\left(\bm{\mu},\frac{\bm{W}\bm{W}^T}{u_2}+\frac{\sigma^2\bm{I}}{u_1}\right).
\end{align}
Then we can determine the probability density function of $\bm{x}$ by integrating out the scale variables $u_1$ and $u_2$ in the joint probability density function: $p(\bm{x},u_1,u_2)$:
\begin{align}
    p(\bm{x})&=\int_0^{\infty}\int_0^{\infty}p(\bm{x}|u_1,u_2)p(u_1)p(u_2)\mathrm{d}u_1\mathrm{d}u_2\nonumber\\
    &=\int_0^{\infty}\int_0^{\infty}N\left(\bm{\mu},\frac{\bm{W}\bm{W}^T}{u_2}+\frac{\sigma^2\bm{I}}{u_1}\right)\cdot \mathrm{Ga}\left(\frac{\nu}{2},\frac{\nu}{2}\right)\cdot \mathrm{Ga}\left(\frac{\nu}{2},\frac{\nu}{2}\right)\mathrm{d}u_1\mathrm{d}u_2,
\end{align}
which cannot be recognized as the probability density function of any multivariate $t$-distributions.

This result is intuitive and can be explained using the generated data. To illustrate how a C\&L $t$-model (8) fits a data set, we generate 10,000 data points from the corresponding generative process:
\begin{equation}
    \left\{\begin{aligned}
      &\bm{x}|\bm{z},u_1 \sim N\left(\bm{W}\bm{z}+\bm{\mu},\frac{\sigma^2\bm{I}}{u_1}\right),\\
    &\ \bm{z}|u_2 \ \ \sim N\left(\bm{0},\frac{\bm{I}}{u_2}\right),\\
    &\ \ u_1 \quad \, \sim \mathrm{Ga}\left(\frac{\nu}{2},\frac{\nu}{2}\right),\\
    &\ \ u_2 \quad \, \sim \mathrm{Ga}\left(\frac{\nu}{2},\frac{\nu}{2}\right).
    \end{aligned}\right.
\end{equation}
In this example, we let the degrees of freedom $\nu_1=\nu_2=3$, the dimension of $\bm{x}$ be $q=2$ and the dimension of $\bm{z}$ be $d=1$, with the transformation matrix $\bm{W}=(2,1)^T$. The scale parameter of the isotropic errors $\sigma^2$ will be set to different values (0.5 and 2), to compare the effects to the overall data structure. To avoid the overall graphical structure to be dominated by a few extreme outliers, 
Figure 1 only shows the region with both coordinates inside $[-100,100]$. 

\begin{figure}[t]
\centering
\includegraphics[width=12cm]{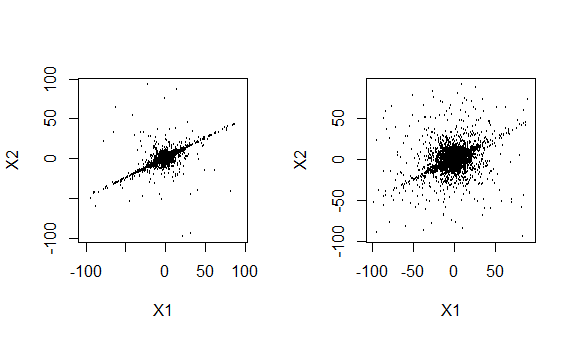}
\caption{Sample data generations from C\&L $t$-models (Left: $\sigma^2=0.5$; Right: $\sigma^2=2$)}
\end{figure}

As we can see from Figure 1, most of the outliers will be only extreme in either the principal axes or the direction orthogonal to the principal axes. It is because most of the latent vectors $\bm{z}$ (scalars in this example) are concentrated in the center of the latent space, and the isotropic noises $\bm{\varepsilon}$ are independent of $\bm{z}$. Therefore, the probability that both the data variable and latent variable take extreme values in both spaces is extremely small (it can only happen when $u_1, u_2$ are simultaneously large enough). This also explains why the marginal distribution of the data $\bm{x}$ will not be elliptical.  

Next, Figure 2 shows 10,000 data points from the generative process of Conditional $t$-models (12), under the same setting as above:
\begin{equation}
    \left\{\begin{aligned}
    &\bm{x}|\bm{z},u \sim N\left(\bm{W}\bm{z}+\bm{\mu},\frac{\sigma^2\bm{I}}{u}\right),\\
    &\ \ \bm{z} \quad \, \sim N\left(\bm{0},\bm{I}\right),\\
    &\ \ u \quad \, \sim \mathrm{Ga}\left(\frac{\nu}{2},\frac{\nu}{2}\right).
    \end{aligned}\right.
\end{equation}

\begin{figure}[t]
\centering
\includegraphics[width=12cm]{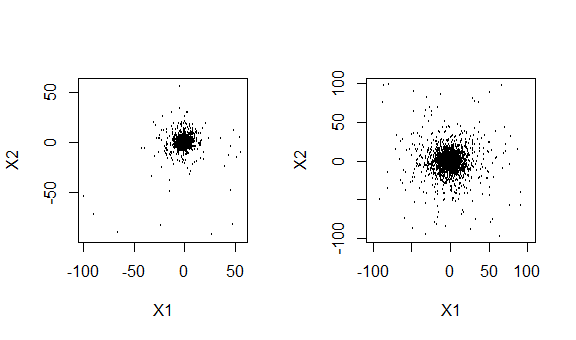}
\caption{Sample data generation from Conditional $t$-models (Left: $\sigma^2=0.5$; Right: $\sigma^2=2$)}
\end{figure}
Intuitively speaking, Conditional $t$-models only assume that the existence of outliers comes from the data space rather than latent space. As a consequence, we can see from Figure 2 that the marginal distribution of the data is more spherical than before. This observation is reasonable from the generative process and distributional assumptions. Using the similar intuitive explanation from \cite{bishop2006pattern}, we can think of the distribution $p(\bm{x})$ as being defined by taking an isotropic multivariate $t$-distributed ``spray can" and moving it across the principal subspace (in our example, one-dimensional principal axis) spraying
spherical ink with density determined by $p(\bm{x}|\bm{z})\sim t_{\nu}(\bm{Wz}+\bm{\mu},\sigma^2\bm{I})$ and weighted by the prior distribution $p(\bm{z})\sim N(\bm{0},\bm{I})$. The accumulated ink density gives rise to a distribution representing the marginal density $p(\bm{x})$. The normal assumption of $\bm{z}$ means that most of latent data will concentrate in the center, and we can only move the ``spray can" within a very small region along the principal subspace. However, the generative distribution $p(\bm{x}|\bm{z})$ is also spherical with center $\bm{\mu}$, but relatively widespread due to the fat tail property of $t$-distributions, the overall shape of the marginal density $p(\bm{x})$ will be dominated by $p(\bm{x}|\bm{z})$, which is spherical. 

\subsubsection{Marginal $t$-Distributed Models}

As we can see from Section 3.1.2, the resulting conditional distribution $p(\bm{x}|\bm{z})$ for Marginal $t$-models is different from that in C\&L $t$-models: $t_v(\bm{W}\bm{z}+\bm{\mu},\sigma^2\bm{I})$. In fact, the conditional distribution $p(\bm{x}|\bm{z})$ for C\&L $t$-models can be obtained by incorrectly integrating out the marginal distribution $p(u)$ instead of $p(u|\bm{z})$:
\begin{align}
    t_v(\bm{W}\bm{z}+\bm{\mu},\sigma^2\bm{I})&=\int^\infty_0 p(\bm{x}|\bm{z},u)p(u)\mathrm{d}u=\int^\infty_0 N\left(\bm{W}\bm{z}+\bm{\mu},\frac{\sigma^2\bm{I}}{u}\right)\cdot \mathrm{Ga}\left(\frac{\nu}{2},\frac{\nu}{2} \right)\mathrm{d}u.
\end{align}
Compared to $t_v(\bm{W}\bm{z}+\bm{\mu},\sigma^2I)$, the correct conditional distribution of $\bm{x}|\bm{z}$ has a larger degrees of freedom $\nu+d$, and a scale matrix $\frac{\nu+\bm{z}^T\bm{z}}{\nu+d}\sigma^2\bm{I}$ involving the observed value of latent variable $\bm{z}$. Interestingly, noticing that $\bm{z}\sim t_{\nu}(\bm{0},\bm{I})$, the correct expected value of the scale matrix $\frac{\nu+\bm{z}^T\bm{z}}{\nu+d}\sigma^2\bm{I}$ converges to $\sigma^2\bm{I}$ as the degrees of freedom $\nu \rightarrow \infty$:
\begin{equation}
    \mathbb{E}\left[\frac{\nu+\bm{z}^T\bm{z}}{\nu+d}\sigma^2\bm{I} \right]=\frac{\nu+\frac{\nu}{\nu-2}d}{\nu+d}\sigma^2\bm{I}\rightarrow \sigma^2\bm{I},
\end{equation}
as $\mathbb{E}(\bm{z}^T\bm{z})=\mathrm{tr}\left(\frac{\nu}{\nu-2}\bm{I}\right)+\bm{\mu}_z^T\bm{\mu}_z= \frac{\nu}{\nu-2}d \rightarrow d$.

Another way to look at Marginal $t$-models is to express $\bm{x}$ under a generative process. After we determined a $d$-dimensional latent variable, a $q$-dimensional data variable will be generated from an affine transformation of $\bm{z}$ plus a $q$-dimensional isotropic Gaussian noise $\bm{\varepsilon}$. The noise vector $\bm{\varepsilon}$ and latent variable vector $\bm{z}$ are jointly multivariate $t$-distributed. Since $(\bm{x},\bm{z})^T$ is an affine transformation of $(\bm{\varepsilon},\bm{z})^T$, $\bm{x}$ and $\bm{z}$ are also jointly multivariate $t$-distributed:
\begin{equation*}
    \bm{x}=\bm{W}\bm{z}+\bm{\mu}+\bm{\varepsilon},
\end{equation*}
\begin{equation}
    \left[\begin{array}{l}{\bm{\varepsilon}} \\ {\bm{z}}\end{array}\right] \sim t_{\nu}\left(\left[\begin{array}{l}{\bm{0}} \\ {\bm{0}}\end{array}\right],\left[\begin{array}{cc}{\sigma^2\bm{I}} & {\bm{0}} \\ {\bm{0}} & {\bm{I}}\end{array}\right]\right),
\end{equation}
\begin{equation*}
    \left[\begin{array}{l}{\bm{x}} \\ {\bm{z}}\end{array}\right] \sim t_{\nu}\left(\left[\begin{array}{l}{\bm{\mu}} \\ {\bm{0}}\end{array}\right],\left[\begin{array}{cc}{\bm{W}\bm{W}^T+\sigma^2\bm{I}} & {\bm{W}} \\ {\bm{W}^T} & {\bm{I}}\end{array}\right]\right).
\end{equation*}
In the standard PPCA models (6), the noise $\bm{\varepsilon}$ and latent variable $\bm{z}$ will be jointly Gaussian distributed. Since uncorrelation implies independence for joint Gaussian random variables, we know that the higher-dimensional noise will be generated independently from the value $\bm{z}$ we observed in the latent space. However, uncorrelation in joint $t$-distributions does not imply independence. To see this, assume a $q$-dimensional random vector $\bm{x}=(\bm{x}_1,\bm{x}_2)^T$ is $t$-distributed and $\bm{x}_1$ and $\bm{x}_2$ are uncorrelated, the kernel of the joint density can be written as:
\begin{align}
    &\left\{1+\frac{1}{\nu}\left(\left[\begin{array}{l}
    {\bm{x}_{1}} \\
    {\bm{x}_{2}}
    \end{array}\right]-\left[\begin{array}{l}
    {\bm{\mu}_{1}} \\
    {\bm{\mu}_{2}}
    \end{array}\right]\right)^{T}
    \left[\begin{array}{cc}{\bm{\Sigma}_{1 1}} & {\bm{0}} \\
    {\bm{0}} & {\bm{\Sigma}_{2 2}}
    \end{array}\right]^{-1}
    \left(\left[\begin{array}{l}
    {\bm{x}_{1}} \\
    {\bm{x}_{2}}
    \end{array}\right]-\left[\begin{array}{l}
    {\bm{\mu}_{1}} \\
    {\bm{\mu}_{2}}
    \end{array}\right]\right)\right\}^{-\frac{\nu+q}{2}}\nonumber\\
    &=\left\{1+\frac{1}{\nu}\left[(\bm{x}_1-\bm{\mu}_1)^T\bm{\Sigma}_{11}^{-1}(\bm{x}_1-\bm{\mu}_1)+(\bm{x}_2-\bm{\mu}_2)^T\bm{\Sigma}_{22}^{-1}(\bm{x}_2-\bm{\mu}_2) \right]\right\}^{-\frac{\nu+q}{2}},
\end{align}
which cannot be factorized into two products of densities of $t$-distributions. Therefore, for Marginal $t$-models, the latent variable $\bm{z}$ and the noise $\bm{\epsilon}$ in data space are uncorrelated, but dependent.
\begin{figure}[t]
\centering
\includegraphics[width=12cm]{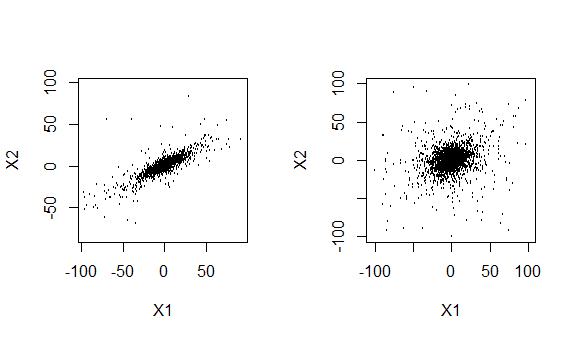}
\caption{Sample data generation from Marginal $t$-models (Left: $\sigma^2=0.5$; Right: $\sigma^2=2$)}
\end{figure}

To show how a Marginal $t$-model (13) fits a data set, we generate 10,000 data points from the corresponding generative process, under the same settings as before:
\begin{equation}
    \left\{\begin{aligned}
    &\bm{x}|\bm{z},u \sim N\left(\bm{W}\bm{z}+\bm{\mu},\frac{\sigma^2\bm{I}}{u}\right),\\
    &\ \bm{z}|u \ \ \sim N\left(\bm{0},\frac{\bm{I}}{u}\right),\\
    &\ \ u \quad \, \sim \mathrm{Ga}\left(\frac{\nu}{2},\frac{\nu}{2}\right).
    \end{aligned}\right.
\end{equation}
The discussion, along with Figure 3, provides insight about how Marginal $t$-models capture the outliers. Because the data variable $\bm{x}$ and latent variable $\bm{z}$ share the common scale variable $u$, they are strongly dependent on each other. When $u$ takes small values, both $\bm{x}|\bm{z},u \sim N(\bm{W}\bm{z}+\bm{\mu},\frac{\sigma^2\bm{I}}{u})$ and $\bm{z}|u \sim N(\bm{0},\frac{\bm{I}}{u})$ will simultaneously have large a variance. Statistically speaking, when the lower-dimensional latent variable is likely to be an outlier in the latent space, the corresponding observation is also likely to be an outlier in the data space. Furthermore, it is obvious that the marginal distributions are both elliptical given different variances of noise, which is consistent with the fact that the marginal distribution of $\bm{x}$ is a multivariate $t$-distribution.

\section{An MCEM Algorithm for Conditional and Latent $t$-Models}
An EM algorithm for Marginal $t$-models has been proposed by \cite{archambeau2006robust} so we will not repeat it in this paper. In this section, we will present an Monte Carlo expectation-maximization algorithm (MCEM) for C\&L $t$-models, where all the considered conditional expectations no longer have closed forms as they in Marginal $t$-models. Therefore, the standard EM algorithms will not work in this case, instead, we can use Monte Carlo methods in the E-step to tackle this issue. The key idea is to first jointly sample $(u_1,u_2,\bm{z})$ conditioning on $\bm{x}$ for each $n=1 \cdots N$, and then compute each expectation using its Monte Carlo estimate. Recall the probabilistic structure of C\&L $t$-models (8):
\begin{equation}
    \left\{\begin{aligned}
      &\bm{x}|\bm{z},u_1 \sim N\left(\bm{W}\bm{z}+\bm{\mu},\frac{\sigma^2\bm{I}}{u_1}\right),\\
    &\ \bm{z}|u_2 \ \ \sim N\left(\bm{0},\frac{\bm{I}}{u_2}\right),\\
    &\ \ u_1 \quad \, \sim \mathrm{Ga}\left(\frac{\nu_1}{2},\frac{\nu_1}{2}\right),\\
    &\ \ u_2 \quad \, \sim \mathrm{Ga}\left(\frac{\nu_2}{2},\frac{\nu_2}{2}\right),
    \end{aligned}\right.
\end{equation}
where the parameters are $(\bm{W},\bm{\mu},\sigma^2,\nu_1,\nu_2)$. Since the marginal distribution of the data $\bm{x}$ is not any recognized distribution, we use an EM algorithm to find the maximum likelihood estimates. The latent variable $\bm{z}$ and scale variable $u_1$, $u_2$ are all treated as ``missing data", and the complete data will comprise those and the observed data $\bm{x}$. The complete log-likelihood is:
\begin{equation}
     L_c = \sum_{n=1}^N \log [p(\bm{x}_n, \bm{z}_n, u_{1n}, u_{2n})]=\sum_{n=1}^N \log[p(\bm{x}_n|\bm{z}_n,u_{1n})p(\bm{z}_n|u_{2n})p(u_{1n})p(u_{2n})].
\end{equation}
% \begin{align}
%     L_c &= \sum_{n=1}^N \log [p(\bm{x}_n, \bm{z}_n, u_{1n}, u_{2n})]\nonumber  \\
%     &= \sum_{n=1}^N \log [p(\bm{x}_n|\bm{z}_n, u_{1n})p(\bm{z}_n|u_{2n})p(u_{1n})p(u_{2n})].
% \end{align}
where (30) can be directly deduced by the direct acyclic graph of C\&L $t$-models. 

In the E-step, we need to find the conditional expectation of the complete log-likelihood $\langle L_c \rangle$ conditioning on the observed data $\bm{x}_n$'s. Substituting the expressions for $p(\bm{x}_n|\bm{z}_n, u_{1n}),\ p(\bm{z}_n|u_{2n}),\ p(u_{1n})\ \mathrm{and}\ p(u_{2n})$ into (30), we obtain $\langle L_c \rangle$:
\begin{align}
    \langle L_c \rangle =& -\sum_{n=1}^N \bigg [\frac{q}{2}\log \sigma^2 +\frac{\langle u_{1n}\rangle}{2\sigma^2}(\bm{x}_n-\bm{\mu})^T(\bm{x}_n-\bm{\mu})-\frac{1}{\sigma^2}\langle u_{1n}\bm{z}_n\rangle^T\bm{W}^T(\bm{x}_n-\bm{\mu})
    \nonumber \\
    &+\frac{1}{2\sigma^2}\mathrm{tr}(\bm{W}^T\bm{W}\langle u_{1n}\bm{z}_n\bm{z}_n^T\rangle)-\frac{\nu_1}{2} \left(\log \frac{\nu_1}{2}+\langle \log u_{1n}\rangle-\langle u_{1n}\rangle \right)+\log \Gamma \left(\frac{\nu_1}{2} \right) \nonumber \\
    &- \frac{\nu_2}{2}\left(\log \frac{\nu_2}{2}+\langle \log u_{2n}\rangle-\langle u_{2n}\rangle \right)+\log \Gamma \left(\frac{\nu_2}{2}\right) \bigg]+\mathrm{constant}.
\end{align}
Since all the expectation terms have no closed forms, we instead sample $B$ times from $(u_1,u_2,\bm{z})|\bm{x}$ and compute the Monte Carlo estimates. Here, we propose an approach using the Gibbs sampler, since each conditional distribution can be easily determined. To run the Gibbs sampler, we first give an initial value $(u_1^{(0)},u_2^{(0)},\bm{z}^{(0)})$, and cycle through the following loop. Assume the current $k^{\text{th}}$ iteration is $(u_1^{(k)},u_2^{(k)},\bm{z}^{(k)})$, then we can sample the next iteration from the following distributions:
\begin{align}
    u_1^{(k+1)}|u_2^{(k)},\bm{z}^{(k)},\bm{x} &\sim \mathrm{Ga} \left(\frac{\nu_1+q}{2},\frac{\nu_1+||\bm{x}-\bm{W}\bm{z}^{(k)}-\bm{\mu}||^2/\sigma^2}{2}  \right),\\
    u_2^{(k+1)}|u_1^{(k+1)},\bm{z}^{(k)},\bm{x} &\sim \mathrm{Ga} \left(\frac{\nu_2+d}{2},\frac{\nu_2+||\bm{z}^{(k)}||^2}{2}  \right),\\
    \bm{z}^{(k+1)}|u_1^{(k+1)},u_2^{(k+1)},\bm{x} &\sim N\left(\bm{M}^{-1}_{(k+1)}\bm{W}^T(\bm{x}-\bm{\mu}),\frac{\sigma^2}{u_1^{(k+1)}}\bm{M}^{-1}_{(k+1)}\right),
\end{align}
where $\bm{M}_{(k+1)}=\bm{W}^T\bm{W}+\frac{\sigma^2u_2^{(k+1)}}{u_1^{(k+1)}}\bm{I}$. Expression (34) can be derived from $\bm{x}|\bm{z},u_1,u_2 \sim N(\bm{Wz}+\bm{\mu},\frac{\sigma^2}{u_1}\bm{I})$ and $\bm{z}|u_1,u_2 \sim N(0,\frac{\bm{I}}{u_2})$, using the Bayes' rule for multivariate normal distributions \citep{bishop2006pattern}. The detailed derivations for (32) and (33) can be found in the Appendix. Then, the Monte Carlo estimates for the $n$th data point can be computed as follows:
\begin{align}
    \langle u_{in} \rangle \approx &\frac{1}{B}\sum_{k=1}^{B} u_{in}^{(k)} \quad (i=1,2),\\
    \langle \log u_{in} \rangle \approx & \frac{1}{B} \sum_{k=1}^{B} \log u_{in}^{(k)}  \quad  (i=1,2),\\
    \langle u_{1n}\bm{z}_n \rangle \approx & \frac{1}{B} \sum_{k=1}^{B} u_{1n}^{(k)}\bm{z}_n^{(k)},\\
    \langle u_{1n}\bm{z}_n\bm{z}_n^T \rangle \approx& \frac{1}{B} \sum_{k=1}^{B} u_{1n}^{(k)}\bm{z}_n^{(k)}\bm{z}_n^{(k)T}.
\end{align}

% \begin{alignat}{2}
%     \langle u_{in} \rangle \approx &\frac{1}{B}\sum_{k=1}^{B} u_{in}^{(k)} \quad (i=1,2), \quad && \quad \langle \log u_{in} \rangle \approx \frac{1}{B} \sum_{k=1}^{B} \log u_{in}^{(k)}  \quad  (i=1,2), \\
%     \langle u_{1n}\bm{z}_n \rangle \approx &\frac{1}{B} \sum_{k=1}^{B} u_{1n}^{(k)}\bm{z}_n^{(k)}, \qquad \qquad &&\langle u_{1n}\bm{z}_n\bm{z}_n^T \rangle \approx \frac{1}{B} \sum_{k=1}^{B} u_{1n}^{(k)}\bm{z}_n^{(k)}\bm{z}_n^{(k)T}. 
% \end{alignat}

In the M-step, we maximize $\langle L_c \rangle$ with respect to $(\bm{W},\bm{\mu},\sigma^2,\nu_1,\nu_2)$, by setting all first order partial derivatives to 0. This leads to the following updating equations:
\begin{align}
    \Tilde{\bm{\mu}}=&\frac{\sum_{n=1}^{N}\left(\langle u_{1n}\rangle\bm{x}_n-\bm{W}\langle u_n \bm{z}_n \rangle \right)}{\sum_{n=1}^{N}\langle u_{1n} \rangle},\\
    \widetilde{\bm{W}}=& \left[\sum_{n=1}^{N}(\bm{x}_n-\Tilde{\bm{\mu}})\langle u_{1n}\bm{z}_n \rangle^T \right]\left[\sum_{n=1}^{N}\langle u_{1n}\bm{z}_n \bm{z}_n^T \rangle \right]^{-1},\\
    \widetilde{\sigma^2}=&\frac{1}{Nq}\sum_{n=1}^{N}\Big[\langle u_{1n}\rangle (\bm{x}_n-\Tilde{\bm{\mu}})^T (\bm{x}_n-\Tilde{\bm{\mu}})-2\langle u_{1n}\bm{z}_n \rangle^T\widetilde{\bm{W}}^T (\bm{x}_n-\Tilde{\bm{\mu}})\nonumber \\
    &+\mathrm{tr}(\widetilde{\bm{W}}^T\widetilde{\bm{W}}\langle u_{1n}\bm{z}_n\bm{z}_n^T\rangle) \Big].
\end{align}
One detail worth noting is that in (40) we update $\bm{W}$ using $\Tilde{\bm{\mu}}$ instead of $\bm{\mu}$, which means this EM algorithm is actually an expectation conditional maximization algorithm (ECM). \cite{dempster1977maximum} have shown that the ECM algorithms belong to the generalised EM algorithms (GEM) and share the same convergence properties as the standard EM algorithms. 

The maximum likelihood estimates of $\nu_1$ and $\nu_2$ can be found by solving the following equations using an one-dimensional linear search:
\begin{align}
    &1+\log \frac{\nu_1}{2}-\psi (\frac{\nu_1}{2})+\frac{1}{N}\sum_{n=1}^N\left(\langle \log u_{1n} \rangle -\langle u_{1n} \rangle \right)=0,\\
    &1+\log \frac{\nu_2}{2}-\psi (\frac{\nu_2}{2})+\frac{1}{N}\sum_{n=1}^N\left(\langle \log u_{2n} \rangle -\langle u_{2n} \rangle \right)=0.
\end{align}
where $\psi(\cdot)$ denotes the digamma function.

The main steps of this algorithm can be summarised as follows:
\begin{itemize}
    \item[] Step 1: Given the current estimates of the parameters $(\bm{W},\bm{\mu},\sigma^2,\nu_1,\nu_2)$, using the Gibbs sampler to take $B$ samples of $(u_1,u_2,\bm{z}|\bm{x})$ .
    \item[] Step 2: Compute the Monte Carlo estimates of the posterior conditional expectations in the E-step.
    \item[] Step 3: Update parameters in the M-step.
    \item[] Step 4: Repeat Step 1-3 until convergence.
\end{itemize}

\section{Simulation Studies}
In this section, we conduct two sets of numerical experiments, one two-dimensional and one 20-dimensional, to compare the performance of C\&L $t$-models, Marginal $t$-models and standard PPCA models. Our goal is to use C\&L $t$-models and Marginal $t$-models to recover the principal axes of the ``true data" when outliers exist. To compare how well the principal subspaces fitted by different PPCA models recover the true principal subspaces, we use the first principal angle between subspaces as the metric. Mathematically, given two subspaces $\mathcal {U},\mathcal {W}$, the first principal angle between $\mathcal {U}$ and $\mathcal {W}$ is defined as \citep{strang1993introduction}:
\begin{equation}
    \theta:=\min \left\{\arccos \left(\left.{\frac {|\langle u,w\rangle |}{\|u\|\|w\|}}\right)\,\right|\,u\in {\mathcal {U}},w\in {\mathcal {W}}\right\}.
\end{equation}
The range of $\theta$ is between $[0,\pi/2]$, where $\theta=\pi/2$ implies the orthogonality between the two subspaces. When both $\mathcal {U}$ and $\mathcal {W}$ have only dimension one, the first principal angle reduces to the ordinary angle between two vectors. 

\begin{figure}[t]
\centering
\includegraphics[width=14cm]{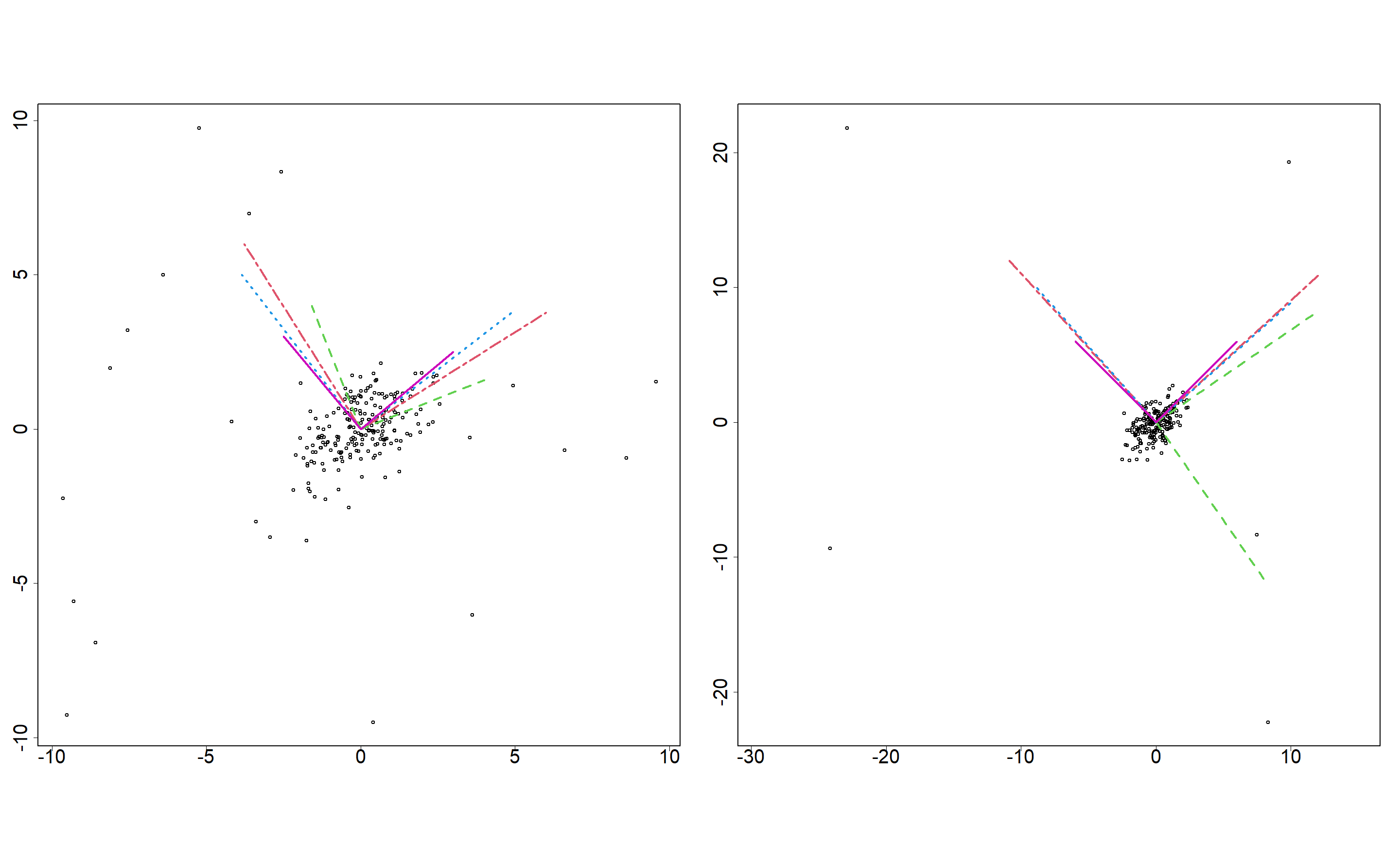}
\caption{One illustration of two-dimensional experiments 2A\&2B: principal axes fitted by the C\&L $t$-model (dark blue, dotted), the Marginal $t$-model (red, twodash), and the standard PPCA model (green, dashed). True principal axes (purple, solid). (Left: 2A: 20 outliers $\overset{iid}{\sim} U[-10,10]^2$; Right: 2B: 5 outliers $\overset{iid}{\sim} U[-25,25]^2$)}
\end{figure}

\begin{figure}[t]
\centering
\includegraphics[width=14cm]{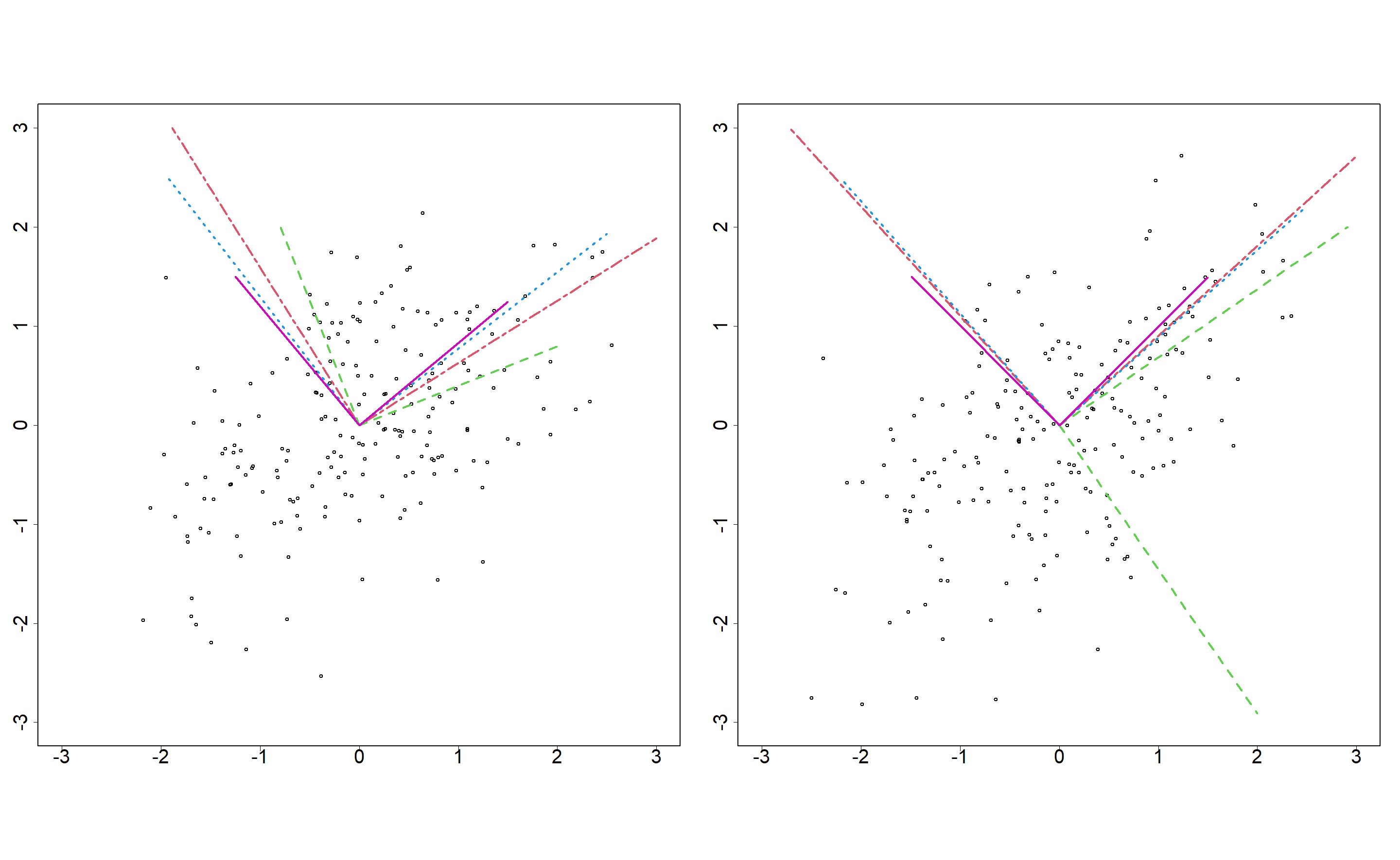}
\caption{The zoom in view of Figure 4 (only shows the region of $[-3,3]^2$)}
\end{figure}
First, we conduct the two-dimensional experiments. In each simulation we simulate $200$ random samples from a two-dimensional normal distribution ($D=2$) with variance $\sigma^2=1$ for each component and correlation coefficient $\rho=0.5$, and the principal axes of this group of data is treated as the ``true data" structure. Then, we add outliers in two different ways in two experiments and run simulations for both. In the first experiment, we generate 20 outliers where each one follows a uniform distribution $[-10,10]^2$, and we denote it as Experiment 2A. The magnitude of the outliers is moderate, which might represent that the outliers are just extreme values compared to other normal points, but still on a reasonable scale. The second approach is to add only 5 outliers where each one is generated from a uniform distribution $[-25,25]^2$, and we denote it as Experiment 2B. This magnitude is extremely large compared to the variance of the normal data, which might mean that the outliers appear due to mismeasurement. For each of Experiment 2A and 2B, we run 100 simulations, and the expected value and standard error (in the subscripts) of the first principal angle between the true principal axe and the fitted ones from different models are shown in Table 1. Notice that we only evaluate the first principal axe ($d=1$) here as the dimension of data $D=2$. 
\begin{table}[h]
\centering
\begin{tabular}{|c|c|c|c|}
\hline
First principal angle & Standard PPCA & Marginal $t$-model & C\&L $t$-model  \\
\hline
Experiment 2A & $0.529_{(0.046)}$  & $0.037_{(0.003)}$  & $0.058_{(0.016)}$   \\
\hline
Experiment 2B & $0.725_{(0.051)}$  & $0.024_{(0.002)}$  & $0.036_{(0.003)}$   \\
\hline
\end{tabular}
\caption{The average first principal angles and standard errors in Experiment 2A and 2B}
\end{table}

For the two-dimensional experiments 2A and 2B, we can visualize one of the simulations to give an illustration. Figure 4 shows the data structure of the selected two experiments and the estimated principal axes using both the C\&L $t$-model (dark blue dotted lines) and the Marginal $t$-model (red two-dashed lines), and uses green dashed lines to represent the principal axes fitted by the standard PPCA model. The shortest purple solid lines indicate the true principal components, which can be easily derived from eigendecomposition of the sample covariance of the data excluding the outliers. Figure 5 is a zoom in view of Figure 4, showing only the region of $[-3,3]^2$, which is focusing on the fit of ``true data". In addition, the estimation results corresponding to Figure 4 and 5 are shown in Table 2 and 3. 

\begin{table}[h]
\centering
\begin{tabular}{|c|c|c|c|c|}
\hline
Coordinate & True & Standard PPCA & Marginal $t$-model & C\&L $t$-model  \\
\hline
$x_1$ & 0.769  & 0.929 & 0.846 & 0.801   \\
$x_2$ & 0.639  & 0.369 & 0.533 & 0.599  \\  
\hline
Principal angle $\theta$ & N/A  & 0.316 & 0.132 & 0.051  \\
\hline
\end{tabular}
\caption{True and fitted first principal axes for the selected example of Experiment 2A (Figure 4\&5: Left)}
\end{table}

\begin{table}[h]
\centering
\begin{tabular}{|c|c|c|c|c|}
\hline
Coordinate & True & Standard PPCA & Marginal $t$-model & C\&L $t$-model  \\
\hline
$x_1$ & 0.709  & 0.566 & 0.741 & 0.749   \\
$x_2$ & 0.705  & -0.824 & 0.672 & 0.663  \\  
\hline
Principal angle $\theta$ & N/A  & 1.390 & 0.046 & 0.058  \\
\hline
\end{tabular}
\caption{True and fitted first principal axes for the selected example of Experiment 2B (Figure 4\&5: Right)}
\end{table}

To further evaluate the performance of different PPCA methods, particularly in higher dimensions, we next conduct 20-dimensional ($D=20$) experiments using the similar setting as in the two-dimensional experiments 2A and 2B. In each simulation, we simulate $200$ random samples from a 20-dimensional normal distribution with variance $\sigma^2=1$ for each component and correlation coefficient $\rho=0.5$. Next, we add outliers in the same two ways as before in the two experiments. In the first experiment, we generate 20 outliers where each one follows a uniform distribution $[-10,10]^{20}$, and we denote it as Experiment 20A. The second way is to add only 5 outliers where each one is generated from a uniform distribution $[-25,25]^{20}$, and we denote it as Experiment 20B. Different from the two-dimensional experiments, here we do not only evaluate the first principal axes $(d=1)$, but also the two-dimensional ($d=2$) and three-dimensional principal subspaces $(d=3)$. For each of Experiment 20A and 20B, we run 100 simulations, and the expected value and standard error (in the subscripts) of the first principal angle between the true principal subspaces ($d=1,2,3$) and the fitted ones from different models are shown in Table 4 and 5.

\begin{table}[h]
\centering
\begin{tabular}{|c|c|c|c|}
\hline
First principal angle & Standard PPCA & Marginal $t$-model & C\&L $t$-model  \\
\hline
$d=1$ & $0.456_{(0.017)}$  & $0.020_{(0.0004)}$  & $0.022_{(0.0004)}$   \\
\hline
$d=2$ & $0.356_{(0.010)}$  & $0.019_{(0.0004)}$  & $0.021_{(0.0004)}$   \\
\hline
$d=3$ & $0.297_{(0.007)}$  & $0.018_{(0.0004)}$  & $0.021_{(0.0005)}$   \\
\hline
\end{tabular}
\caption{The average first principal angles and standard errors in Experiment 20A with principal subspace dimension $d=1,2,3$}
\end{table}

\begin{table}[h]
\centering
\begin{tabular}{|c|c|c|c|}
\hline
First principal angle & Standard PPCA & Marginal $t$-model & C\&L $t$-model  \\
\hline
$d=1$ & $1.274_{(0.022)}$  & $0.018_{(0.0004)}$  & $0.020_{(0.0004)}$   \\
\hline
$d=2$ & $1.058_{(0.019)}$  & $0.017_{(0.0004)}$  & $0.020_{(0.0004)}$   \\
\hline
$d=3$ & $0.820_{(0.017)}$  & $0.015_{(0.0004)}$  & $0.018_{(0.0005)}$   \\
\hline
\end{tabular}
\caption{The average first principal angles and standard errors in Experiment 20B with principal subspace dimension $d=1,2,3$}
\end{table}

From both the two-dimensional (2A\&2B) and 20-dimensional (20A\&20B) experiments, we can see that the standard PPCA fails to recover the true principal axes accurately, under two different outlier structures. On the other hand, both the C\&L $t$-model and the Marginal $t$-model estimate the principal axes quite close to true ones, with considerably small standard errors. The small value of principal angles and standard errors together imply the robustness of both the C\&L $t$-model and the Marginal $t$-model under the presence of outliers. In the conducted experiments, the C\&L $t$-model and the Marginal $t$-model produce extremely close estimation accuracy and standard error.

However, the benefit of using Marginal $t$-models is twofold, one is the simpler computation, and two the marginal multivariate $t$-distribution structure on the data variable $\bm{x}$. Since the majority of the data points will have an elliptical structure in many cases, this means that even though Marginal $t$-models fail to capture a small amount of outliers accurately, the good fit to most normal points will eventually lead a overall reasonably satisfying result. In terms of C\&L $t$-models, the main drawback is the heavy computation since we introduce a Monte Carlo step into the E-step to compute the expectations. On the model fitting side, C\&L $t$-models will generally have a better fit to the outliers since it does not have strong restrictions to the structure of outliers, and there are two degrees of freedom adjusting the overall robustness in different directions. However, since the marginal distribution of data $\bm{x}$ is neither longer multivariate $t$ nor elliptical, C\&L $t$-models might lose some goodness of fit to the main part of the data, which are not outliers. This is why C\&L $t$-models work worse than Marginal $t$-models in some cases, even though they capture outliers better from the distribution point of view. 

\section{Conclusion}
This paper proposes several robust probabilistic principal component analysis models based on multivariate $t$-distributions and their corresponding hierarchical structures. It also gives detailed derivations of those equivalence between models, and clarifies the results of \cite{archambeau2006robust} and several follow-up papers. Interpretations and comparison between models are discussed in detail, through both theory and simulations. 

As we see in the simulation study, both C\&L $t$-models and Marginal $t$-models work reasonably well in some cases, but Marginal $t$-models are more computationally efficient. Therefore, we would still suggest to use the existing robust PPCA algorithm, the Marginal $t$-models, as a robust PPCA model for most applications. However, when one is going to modify the model and derive the corresponding EM algorithm, care is needed when transforming from the multivariate $t$-model to the corresponding hierarchical Normal-Gamma model.

\section*{Appendix}
In this appendix, we will derive (32) and (33) involved in the E-step of the MCEM algorithm for C\&L $t$-models in Section 4. 
\begin{proposition}
    The conditional distribution $u_1|u_2, \bm{z},\bm{x}$ is a Gamma distribution (32):
    $$u_1|u_2,\bm{z},\bm{x} \sim \mathrm{Ga} \left(\frac{\nu_1+q}{2},\frac{\nu_1+||\bm{x}-\bm{Wz}_n-\bm{\mu}||^2/\sigma^2}{2}  \right).
    $$
\end{proposition}

\begin{proof}
Recall that the joint density of $(\bm{x},\bm{z},u_1,u_2)$ in C\&L $t$-models can be decomposed from the DAG to:
\begin{equation}
    p(\bm{x},\bm{z},u_1,u_2)=p(\bm{x}|\bm{z},u_1)p(\bm{z}|u_2)p(u_1)p(u_2).
\end{equation}
It is intuitive that $u_1$ and $\bm{z}$ are independent, since the distribution of $\bm{z}$ will be completely determined by $u_2$, which is independent from $u_1$ by assumption. We can also show this formally:
\begin{align}
    p(\bm{z},u_1)&= \iint p(\bm{x}|\bm{z},u_1)p(\bm{z|u_2})p(u_1)p(u_2)\mathrm{d}\bm{x}\mathrm{d} u_2= p(u_1) p(\bm{z}),
\end{align}
which implies $u_1|\bm{z}\overset{d}{=}u_1 \sim \mathrm{Ga} (\frac{\nu_1}{2},\frac{\nu_1}{2})$. Combining with $(\bm{x}|\bm{z},u_1)\sim N(\bm{Wz}+\bm{\mu},\frac{\sigma^2}{u_1}\bm{I})$, we can use the normal-gamma conjugacy described in (3) to obtain the distribution of $(u_1|\bm{z},\bm{x})$:
\begin{equation}
    u_1|\bm{z},\bm{x} \sim \mathrm{Ga} \left(\frac{\nu_1+q}{2},\frac{\nu_1+||\bm{x}-\bm{Wz}_n-\bm{\mu}||^2/\sigma^2}{2}  \right).
\end{equation}
Next, we show $(u_1|u_2,\bm{z},\bm{x})\overset{d}{=}(u_1|\bm{z},\bm{x})$, which means $u_1$ and $u_2$ are independently conditional on $(\bm{z},\bm{x})$. Because $u_2$ determines $\bm{z}$ only, all the information of $u_2$ is known as long as we know $\bm{z}$. Formally,
\begin{align}
    p(u_1|u_2,\bm{z},\bm{x})&=\frac{p(\bm{x}|\bm{z},u_1)p(\bm{z}|u_2)p(u_1)p(u_2)}{\int p(\bm{x}|\bm{z},u_1)p(\bm{z}|u_2)p(u_1)p(u_2)\mathrm{d}u_1}= p(u_1|\bm{z},\bm{x}).
\end{align}
Therefore, the required conditional distribution $u_1|u_2,\bm{z},\bm{x}$ also takes the same form as (47):
\begin{equation}
    u_1|u_2,\bm{z},\bm{x} \sim \mathrm{Ga} \left(\frac{\nu_1+q}{2},\frac{\nu_1+||\bm{x}-\bm{Wz}_n-\bm{\mu}||^2/\sigma^2}{2}  \right),
\end{equation}
as proposed in (32).
\end{proof}

\vspace{0.5cm}

\begin{proposition}
    The conditional distribution $u_2|u_1, \bm{z},\bm{x}$ is a Gamma distribution (33):
    $$u_2|u_1,\bm{z},\bm{x} \sim \mathrm{Ga} \left(\frac{\nu_2+d}{2},\frac{\nu_2+||\bm{z}||^2}{2}  \right).
    $$
\end{proposition}

\begin{proof}
    Since $u_2 \sim \mathrm{Ga} (\frac{\nu_2}{2},\frac{\nu_2}{2})$, and $\bm{z}|u_2\sim N(\bm{0},\frac{\bm{I}}{u_1})$, we can use the normal-gamma conjugacy (3) to obtain the distribution of $u_2|\bm{z}$:
\begin{equation}
    u_2|\bm{z} \sim \mathrm{Ga} \left(\frac{\nu_2+d}{2},\frac{\nu_2+||\bm{z}||^2}{2}  \right).
\end{equation}
Next we show that $(u_2|u_1,\bm{z},\bm{x})\overset{d}{=}(u_2|\bm{z})$, which means $u_2$ and $(u_1,\bm{x})$ are independently conditional on $\bm{z}$. Since $u_2$ only directly determines $\bm{z}$, it means that as long as we know $\bm{z}$, there will be no extra information about $u_2$ even though we know $\bm{x}$ and/or $u_1$ additionally. We can also show this formally:
\begin{align}
    p(u_2|u_1,\bm{z},\bm{x})&=\frac{p(\bm{x}|\bm{z},u_1)p(\bm{z}|u_2)p(u_1)p(u_2)}{\int p(\bm{x}|\bm{z},u_1)p(\bm{z}|u_2)p(u_1)p(u_2)\mathrm{d}u_2}= p(u_2|\bm{z}).
\end{align}
Therefore, the required conditional distribution $u_2|u_1,\bm{z},\bm{x}$ also takes the same form as (50):
\begin{equation}
    u_2|u_1,\bm{z},\bm{x} \sim \mathrm{Ga} \left(\frac{\nu_2+d}{2},\frac{\nu_2+||\bm{z}||^2}{2}  \right).
\end{equation}
as proposed in (33).
\end{proof}

\bibliography{Reference}

\end{document}